\documentclass[12pt,fleqn]{article}
\title{Deterministic Theories}
\author{Hans P. Halvorson, JB Manchak, James Owen Weatherall}
\date{\today}
\usepackage{tikz}
\usetikzlibrary{decorations.pathmorphing} 
\usepackage{tikz-cd}
\usepackage{amsmath}
\usepackage{amsfonts}
\usepackage{amsthm}
\newtheorem{prop}{Proposition}

\theoremstyle{definition}

\newtheorem{ps}{Proposal}
\newtheorem*{fd}{Full Determinism}
\newtheorem*{example}{Example}
\usepackage{amssymb}
\usepackage{graphicx}
\newcommand{\de}{\sim _{\mathrm{r}}}
\newcommand{\qe}{\sim _{\mathrm{q}}}
\usepackage[utf8]{inputenc}
\usepackage[T1]{fontenc}
\usepackage[colorlinks=false]{hyperref}
\usepackage[style=authoryear, backend=biber, natbib=true, doi=true]{biblatex} 
\usepackage{marvosym}

\newcommand{\qs}{\Sigma _{\mathrm{q}}} \newcommand{\ds}{\Sigma _{\mathrm{r}}}

\begin{document}


\maketitle




\begin{abstract}
  Determinism is (roughly) the thesis that the past determines the
  future. But efforts to define it precisely have exposed deep
  methodological disagreements. Standard possible-worlds formulations
  of determinism presuppose an ``agreement'' relation between worlds,
  but this relation can be understood in multiple ways, none of which
  is particularly clear.  We critically examine the proliferation of
  definitions of determinism in the recent literature, arguing that
  these definitions fail to deliver clear verdicts about actual
  scientific theories. We advocate a return to a formal approach, in
  the logical tradition of Carnap, that treats determinism as a
  property of scientific theories, rather than an elusive metaphysical
  doctrine.

  We highlight two key distinctions: (1) the difference between
  qualitative and ``full'' determinism, as emphasized in recent
  discussions of physics and metaphysics, and (2) the distinction
  between weak and strong formal conditions on the uniqueness of world
  extensions. We argue that defining determinism in terms of
  metaphysical notions such as haecceities is unhelpful, whereas
  rigorous formal criteria such as Belot's D1 and D3 offer a tractable
  and scientifically relevant account. By clarifying what it means for a
  theory to be deterministic, we set the stage for a fruitful
  interaction between physics and metaphysics. \end{abstract}

\section{Determinism: Capturing the intuition}

The thesis of determinism seems easy enough to state:
\begin{quote}
  \emph{Determinism:} The past determines the future.
\end{quote}
Unfortunately, stated this way, the thesis is uninformative, because the term ``determines'' on the right is just a variant of the
term we are trying to define. Can we do better than this?

For many philosophers over the past fifty years, that the solutions has been to invoke possible worlds to sharpen the thesis:\footnote{Of course, possible worlds talk is most commonly associated with \citet{kripke} and, in a different way, \citet{lewis}.  See \citet{loux} for a helpful early anthology showing how these ideas spread through metaphysics and epistemology in the early 1970s.  It enters the philosophy of physics, and especially the analysis of determinism, most prominently with \citet{primer}.  Much more recently, there has been a movement away from ``modality-first'' approaches to metaphysics \citep{Dorr+etal,Fritz,sider}, including among philosophers of physics \citep{cudek+arana}.  Nonetheless, the argot persists.}
\begin{quote}
  \emph{Determinism:} For any two possible worlds $W,W'$, if $W$ and
  $W'$ agree on the past, then $W$ and $W'$ agree on the
  future. \end{quote} This possible-worlds definition gives a
quasi-mathematical gloss to the word ``determines'': there is a
well-defined function from past segments to possible worlds. It seems
like we have made good progress already in cashing out the notion of
determinism. We have replaced the opaque word ``determines'' with what
appears to be a statement about existing things and a clear relation
(``agree'') between them.

But there is a problem. The relation ``agree'' has turned out to be a
point of great contention among philosophers. To illustrate with a
simple example:
\begin{quote} (Q) Let $W$ be a world with one particle on the left and
  an identical particle on the right. Let $W'$ be the world in which
  the two particles are swapped. Do $W$ and $W'$ agree? \end{quote}
Indecision about Q can infect our judgment about whether
determinism holds in concrete examples. For example, we can imagine
that $W$ and $W'$ are exactly identical at all times leading up to
now. If $W$ and $W'$ are judged to agree at the present time, then
determinism holds; if $W$ and $W'$ are judged not to agree at the
present time, then determinism fails. It seems that the fairly clear
notion of determinism has become clouded by esoteric questions
about what it means for possible worlds to agree.

One's initial reaction might be that these kinds of examples only show
that the notion of determinism was less clear than we imagined it
was. But what are we to say about the fact that scientists seem
confident in their judgments about whether physical theories are
deterministic? For example, there is a genuine difference between
deterministic and stochastic equations of motion. Or, perhaps more
controversially, but still to the point: it seems quite clear that
there is a difference vis-a-vis determinism between quantum mechanics
and the theories that went before it.  Scientists appear to be able
make these distinctions reliably without having a theory of when
possible worlds ``agree''.  What is behind those judgments?

Our basic claim in this paper is that whether or not a physical theory
is deterministic can be seen by analyzing the mathematical structure
of the theory.  Talk of possible worlds is helpful only insofar as it
sometimes functions as a substitute for talk of models of a theory;
for most purposes, it introduces more problems than it solves.  And we
do not need a different metaphysical framework to replace Lewisian
possible worlds, either.  Mathematics is perfectly adequate for this
task.

Our arguments run against the grain of the past three decades of work
on determinism in physics (and metaphysics), which has neared
consensus that there can be no adequate ``formal'' definition of
determinism.  A theory, it is said, can only be deterministic or
indeterministic under an ``interpretation''.  This is mistaken, at
least as usually intended, both because it relies on confused ideas
about ``interpretation'' and because a perfectly adequate formal
criterion of determinism is already available --- and has been for
thirty years.  We will argue that challenges for this criterion that
have previously been discussed do not show that ``interpretation'' is
needed; rather, they highlight the importance of precisely formulating
a theory, as small differences in what we take a theory to be can lead
to different judgments about whether that theory is deterministic.

\section{The fate of formal definitions}

The young Bertrand Russell frequently contrasted the clarity of
nineteenth century mathematics with the opacity of Hegelian
metaphysics. He pointed out that mathematicians had provided sharp
definitions of the traditional number systems as well as general
concepts such as infinite sets and continuous functions.
\begin{quote}
  What is infinity? If any philosopher had been asked for a definition
  of infinity, he might have produced some unintelligible rigmarole,
  but he would certainly not have been able to give a definition that
  had any meaning at all.  Twenty years ago, roughly speaking,
  Dedekind and Cantor asked this question, and, what is more
  remarkable, they answered it. They found, that is to say, a
  perfectly precise definition of an infinite number or an infinite
  collection of things. \citep[p 92]{russell}
\end{quote} Russell then proposed that philosophers should model their
method on that of the nineteenth century mathematicians. Obviously,
Russell's vision was a shaping force in analytic philosophy, beginning
with Carnap's attempts to ``explicate'' the concepts of the natural
sciences \citep[see][]{leitgeb}.

In \emph{Logical Syntax of Language}, Carnap suggests that some deep
and murky metaphysical theses correspond to precise statements about
the structure of scientific theories. In particular, with regard
to determinism, he says:
\begin{quote}
  The opposition between the determinism of classical physics and the
  probability determination of quantum physics concerns a syntactical
  difference in the system of natural laws, that is, of the P-rules of
  the physical language. \citep[p 307]{lsl}
\end{quote}
More concretely, Carnap suggests that the metaphysical doctrine:
\begin{quote}
  Every process is univocally determined by its causes \end{quote}
corresponds to a formal property of a scientific theory:
\begin{quote}
  For every particular physical sentence $\varphi$, there is for any
  time coordinate $t$, which has a smaller value than the time
  coordinate which occurs in $\varphi$, a class $\Gamma$ of particular
  sentences with $t$ as time coordinate, such that $\varphi$ is a
  P-consequence of $\Gamma$. \end{quote} The formal property described
here may seem more opaque than the original metaphysical
doctrine. Nonetheless, unlike the original metaphysical doctrine,
whether a theory is deterministic in the latter sense is something
that could be checked by a mathematician, so long as the relevant
theory has been specified in a mathematically clear fashion.

Even after Carnap's methods fell under attack by Quine, some notable
philosophers continued to think of determinism as a formal property of
scientific theories. J.J.C. Smart claims that:
\begin{quote} A perfectly precise meaning can be given to saying that
  certain theories are deterministic or indeterministic (for example
  that Newtonian mechanics is deterministic, quantum mechanics
  indeterministic), but our talk about actual events in the world as
  being determined or otherwise may be little more than a reflection
  of our faith in prevailing types of physical theory. \citep[p
  294]{smart-free} \end{quote} But change was in the wind, and by the
1970s, one begins to find a different kind of definition of
determinism. The most influential of these ``metaphysical''
definitions of determinism is from David Lewis:\footnote{Lewis credits
  \citet{montague} with a similar definition, but notes that Montague
  does not focus on metaphysical issues. We won't discuss Montague's
  paper in detail, but we suggest that his argument against a
  ``syntactic'' definition of determinism is unconvincing, and is
  based on a false dilemma between syntactic and semantic
  definitions.}
\begin{quote} A system of laws of nature is Deterministic iff no two
  divergent worlds both conform perfectly to the laws of that
  system. Second, a world is Deterministic iff its laws comprise a
  Deterministic system. Third, Determinism is the thesis that our
  world is Deterministic. \citep[p 360]{lewis-work} \end{quote} Lewis'
definition was put front and center in philosophy of physics by
\citet{primer}, and it has ever-since served as the backdrop for the
``hole argument'' in General Relativity
\citep[see][]{butterfield,pooley}. More generally, Lewis'
diverging-worlds conception of determinism informs a wide range of
discussions in analytic philosophy.

What was originally an inclination to do philosophy in a different way
(than the logical positivists) soon developed into theoretical
arguments against formal approaches. For example, \citet{belot-do}
argues that determinism cannot possibly be defined in anything like the way that
Carnap and Smart proposed, because it is not a \emph{formal} property of (mathematical) theories:
\begin{quote} The first point that I would like to make is that
  determinism cannot be a \emph{formal} property of physical
  theories. \citep[p 88, emphasis in original]{belot-do} \end{quote}

Belot's position here is part of a trend among analytic metaphysicians and metaphysically-oriented philosophers of science away from formal definitions. For example, regarding formal definitions of equivalence
of theories, Sider claims that, ``the purely formal approach is a
nonstarter'' \citeyearpar[p 180]{sider}, and, ``purely formal accounts
fail because they entirely neglect meaning'' \citeyearpar[p
181]{sider}. A similar complaint against formal definitions of
theoretical equivalence is voiced by Kevin Coffey:
\begin{quote} The challenges posed by the two puzzles are not unique
  to formal approaches, but I think we should be particularly
  pessimistic about the prospects of formal approaches meeting those
  challenges.  \citep[p 834]{coffey} \end{quote} And Trevor Teitel
argues that formal definitions are of little interest for
philosophical investigations:
\begin{quote} I will investigate various views one might hold about
  the non-mathematical significance of these formal criteria, and
  argue that none is tenable. My tentative conclusion is that formal
  criteria are of limited non-mathematical interest. \citep[p
  4120]{teitel-eq} \end{quote}

What justifies this widespread rejection of formal methods? Sider, Coffey, and Teitel make different arguments from Belot, and from one another.  But they also put great weight on the importance of \emph{interpretation} or, as Sider says, \emph{meaning}. We will focus on Belot's argument against formal definitions of determinism.  When we see why that argument fails, it will be clear why we would reject other arguments for the same kind of view.  And then, with formal methods rehabilitated, we will turn to a formal analysis of other arguments in the literature on determinism.

\section{Interpretation is a formal matter}\label{sec:interpretation}

Belot's argument for the claim that determinism cannot be a formal
property of theories involves an example.  He presents a set of
equations --- Maxwell's equations, describing classical
electromagnetism, written in a particular way --- that he claims are
standardly understood to be part of a deterministic theory, but which
may yet also be part of an indeterministic one.  The difference, he
says, comes down to interpretation.
 \begin{quote}
   This completes the argument: determinism cannot be a formal
   property of theories, because the same theory may be deterministic
   or indeterministic, depending on how it is
   interpreted.\citep[p. 88]{belot-do} \end{quote} The argument as
 stated has a suppressed premise: that interpretation is not a formal
 matter. But we deny this premise, twice over. First, the claim that
 interpretation is not a formal matter takes interpretation to involve
 a re-negotiation of the relationship between language and
 reality. Second, and relatedly, the claim that interpretation is not
 a formal matter fails to see that interpretation of a theory very
 often calls for a formal precisification of that theory.

We can make both points by describing Belot's example in more detail. As he set it up, it concerns two theories.  One of these theories is classical electromagnetism, written in terms of a scalar (electric) potential and a vector potential, and satisfying Maxwell's equations.  The second theory concerns the properties of a ``plenum of point particles'' called ``blips'' (88).  Blips are represented by two fields on (Newtonian) spacetime: a scalar field $\varphi$ representing the charge of the blip(s) at each point of spacetime; and a vector field $A$ representing their velocity at each point.  The blip theory is remarkable because it describes the motion of blips using just two equations---equations that happen to be syntactically identical to Maxwell's equations written for a scalar potential $\varphi$ and vector potential $A$.\footnote{See \citet[p. 88]{belot-do} for details, and a fine example of his inimitable style.} Belot argues the first of these theories is deterministic, standardly (and properly) construed.  But the second is indeterministic.  The reason is that a specification of the charges and velocities of blips throughout space at a time cannot uniquely determine the charges and velocities of blips at other times.  Given any  solutions to the blip equations, one can always apply any of a family of transformations to produce another solution that agrees on the initial specification.  This is so even though they involve the same equations as Maxwell's theory.

How could it be that one of these theories is deterministic, but the other is not?  Cognoscenti may guess the answer, but we will postpone giving it.  Instead, consider what Belot takes to be at issue. For him, the difference between the theory is a matter of how they are interpreted.  What does he mean by
``interpretation''? He poses and answers the question directly.
\begin{quote} What \emph{is} an interpretation of a theory? An
  interpretation of a theory is a story that you tell about the
  theory.  \dots An interpretation is a correspondence between bits of
  models of the theory and bits of physical situations: between
  initial value constraints, variables and differential equations on
  the one hand, and instantaneous states, physical entities,
  properties, relations, etc., and laws of nature on the
  other. \citep[p 92]{belot-do} \end{quote} This passage suggests a widespread view
in late twentieth century analytic philosophy, which is that to interpret something like a physical theory is to apply some ``semantic glue'', thereby attaching symbols to their worldly referents. But how does this metaphor actually work?  Belot's articulation should give us
pause, because it involves an odd mix of two very different things.
The first is the mathematical theory of interpretation, developed by
Tarski et al. six decades before, which explicates an interpretation
of a formal language as an assignment of set-theoretic extensions to
meaningless symbols.  Belot's talk of the ``models'' of a theory
invokes this tradition.

But Belot also motions towards a much woolier, and apparently older, sense of ``interpretation'' when he speaks of a correspondence between models and the world.  In doing so, Belot has re-doubled the notion of interpretation.  Recall that in its
Tarskian sense, an interpretation of a theory \emph{is} a model of
that theory. But now Belot speaks of interpretation as, ``a
correspondence between bits of models of a theory and bits of physical
situations''. In this case, the interpretations (models) of a theory
are being interpreted. But what does that mean?  How would we produce or exhibit a ``correspondence between bits of models of the theory and bits of physical situations''?

Belot does not say how this is supposed to go in general.  And when it comes time for him to describe two theories that differ (only?) in interpretation, what we find is not a ``correspondence between bits of models of the theory and bits of physical situations'', but rather some \emph{words}, further elaborating some formal aspects of two theories he wishes us to understand.  The exercise is all carried out on paper, by sketching a correspondence between bits of models of the theory and yet another kind of symbol: words in natural language.  But this is just to layer (formal, symbolic) ``interpretation'' on ``interpretation''---that is, to do the ``redoubling'' we complained about above.  And on reflection, given that he has written a paper, it is hard to see what his sense of interpretation could amount to aside from creating yet
other models, and associating bits of the old models to bits of the new
models.  So in fact, though he motions at something different, what he exhibits as interpretation does not give us some new way of crossing the
word-world barrier; it merely connects things that lie on the
word side, the formal side, to other things on that side.

One might well object here: how does \emph{any} of our linguistic or
symbolic practice have meaning, if all of our ``interpretation''
consists of layering models on models and words upon words?  Surely,
one might say, we do sometimes interpret in the woolier, and perhaps
deeper, sense --- and at least some of our words have meaning.  Of
course we accept this.  But we suggest this line of thought just leads
to a dual perspective on our position, entirely consistent with what
we have said thus far.  From this perspective, \emph{all} of the
formal, mathematical, linguistic, and symbolic structures that we
employ, in ordinary life, in philosophy, and in science, have meaning,
at least to some extent, because of the way they are embedded in our
broader cultural and cognitive processes.  To put the point in terms of the metaphor noted above: there is no point in the paper in which Belot \emph{applies} semantic glue.  The glue -- or rather, whatever processes allow any of our symbolic activity to play the role it does in reasoning, action, and communication -- is already there.

\citet{interpretation} puts the point nicely.
\begin{quote}We don't begin our analysis of scientific theories by
  taking some mysterious equations carved on stone tablets and
  puzzling out what they might mean: theories are born as bearing all
  kinds of semantic or interpretational relationships to our broader
  representational practice. ... [T]he problem is not how to
  comprehend an alien practice, but how to fully understand a practice
  which we already --- at least to some extent --- inhabit. (18)
\end{quote}
Again, we will not attempt to give a story of how meaning works.  Our
point is that there is no line to draw between ``merely'' or
``purely'' formal methods or structures, between truly ``meaningless
symbols'' or marks on a page, and any other symbols we might try to
employ to make our ideas clear or communicate them to others.  The
idea from the formal theory of interpretation that we are giving
meaningless symbols ``meaning'' by mapping them to set-theoretic
structures is an idealization --- as should be obvious, since
\emph{both} sides of the interpretation map consist of rich but
thoroughly formal theories.\footnote{\label{urelement} Some readers will raise an objection at this point about the role of (literal, physical) ``ur-elements'' in formal semantics.  Addressing this thought would take us too far afield and we will leave it to future work.  For now, suffice it to say that we reject the idea that physical things can be ``in'' abstract ones; set theory with ur-elements is best viewed as a formal theory living within ordinary set theory, where some abstract objects have highly suggestive names.}   Likewise, interpretation in the wild
consists of using meaningful symbols to guide reasoning about other
meaningful symbols.  It is all just more of the same.

To illustrate the point, we now return to how Belot can maintain that
electromagnetism is deterministic, while blip theory is not.  Belot
does not fully articulate these theories in a formal way.  But he
describes them with sufficient clarity that one can see why they are
in fact mathematically distinct --- and it is precisely those
mathematical differences that support his judgments about whether they
are deterministic.  The crucial step occurs when, in describing the
theories, he explains when putatively distinct models of
electromagnetism actually correspond to the same physical situations
(described in different terms)---whereas for the blip theory, those
same models correspond to distinct physical situations.  This amounts
to a disambiguation between two possible theories, given by specifying
when models are \emph{isomorphic}, i.e., when they should be seen as
mathematically (and, for the purposes of analyzing the theory,
semantically) equivalent.  According to one disambiguation of these
theories (namely, electromagnetism), there is an isomorphism between
models if they are related by a certain class of transformations known
as \emph{gauge transformation}.\footnote{See \citet{gauge} for a
  discussion of this way of understanding gauge transformations in
  electromagnetism, along with a different way of distinguishing
  standard electromagnetism from its close cousins with different
  isomorphisms.}  Among the gauge transformations are all of the
transformations that related ``distinct'' evolutions of a given
initial configuration of the blip theory.  Meanwhile, according to the
disambiguation of electromagnetism --- or rather, the blip theory ---
gauge transformations are not isomorphisms, and thus count as
(formally) distinct by the lights of that theory. This is what
supports the conclusion that electromagnetism is deterministic but the
blip theory is not.

The important point for our purposes is that everything here is formal:
the description of the theory --- both models and the relationship of
isomorphism between models --- as well as the definition of
determinism in the background. Belot has not escaped from the circle of formalism and
into a realm of inarticulable concreta; he has merely motioned toward
the power of formal tools to capture conceptual distinctions. And we
thank him for this: he shows clearly that what appears to be a single
theory is in fact two, depending on how one spells out the details. We can call
this explication of formal structure ``interpretation'' if we want ---
as Tarski did --- but in that case, interpretation is a formal
matter.

Of course, we have not given a first-order theory then specified a
Tarskian interpretation function.  We operate instead at a level
closer to ordinary mathematics.  But even so, it shows how the salient
activity for ``interpreting'' our theories involves \emph{formal}
work--- in this case, the formal (or mathematical) work of specifying
which models are isomorphic --- because a different specification of
the relationship of isomorphisms gives a different theory, and it
might transform a deterministic theory into an indeterministic one, or
vice versa. We will presently see a more subtle example of where
careful interpretation of a different kind, though still formal, can
disambiguate between theories whose differences are difficult to
detect through more conversational descriptions.

\section{From explication to metaphysical speculation}

With that defense --- and explanation --- of formal methods in mind,
we now return to analyses of determinism downstream of Lewis. There
are two directions in which Lewis' \citeyearpar{lewis-work} definition
of determinism has been developed.  One is the quasi-mathematical
approach introduced by Earman and Butterfield, but then attacked by
Belot.  We return to that below.  First, we will consider the other
direction of development.  On this branch, Lewis' definition has been
suited out with various metaphysical distinctions --- and especially
an elusive distinction between ``qualitative'' and ``full'' agreement
of possible worlds. These metaphysical adumbrations on Lewis'
definition have given rise to a significant literature which might
seem, at first, to evince a rich interaction between physics and
metaphysics.  But as we will presently argue --- bringing formal
methods to bear --- these appearances are misleading.

Consider Hawthorne's \citeyearpar{hawthorne} proposed distinction
between two senses in which the world might be
deterministic.\footnote{A similar distinction can be found in
  \citep{melia}, though not quite as sharply put.}

\begin{quote} \emph{Qualitative Determinism:} For all times $t$, there
  is no possible world which matches this world in its qualitative
  description up to $t$, and which has the same laws of nature as this
  world, but which doesn't match this world in its total qualitative
  description. (239)\end{quote}

\begin{quote} \emph{De Re Determinism:} For all times $t$, there is no
  possible world which matches this world in its de re description up
  to $t$, and which has the same laws of nature as this world, but
  which doesn't match this world in its total de re
  description. (239) \end{quote}

Hawthorne's distinction seems to be nearly ubiquitous in the
literature. For example, \citet{teitel} distinguishes laws that are
qualitatively deterministic from laws that are ``fully''
deterministic. Similarly, \citet{pooley} claims that what he calls
``substantivalist'' General Relativity is deterministic to a lower
degree ``Det2'',\footnote{To be sure, \citet[\S5]{pooley} also argues
  that an anti-haecceitist -- such as himself
  \citep[c.f.][pp.99-103]{pooleyOld} -- would maintain that Det1 and
  Det2 are strictly equivalent, though he also feels the draw of
  examples, such as those discussed in the main text, intended to show
  Det2 to be defective.  We see no value to the qualifier
  ``substantivalist'' here, since what Pooley has in mind is
  apparently just textbook General Relativity.  We will return to this
  point below; for now, we will drop the qualifier.} but not in the
most eminent sense ``Det1''.
\begin{quote}
  \begin{tabular}{l|l|l|l}
    & grade 1 & grade 2 \\ \hline \hline
    Hawthorne & qualitative & de re \\ \hline
    Teitel    & qualitative & full \\ \hline
    Pooley    & Det2        & Det1 \\ \hline
    \end{tabular}
\end{quote}
Each of these authors
indicates that there is some thicker kind of determinism that is not
established by standard proofs, e.g.\, of the initial value problem
for partial differential equations.

Clearly, the terminology here has become variegated.\footnote{Or worse, because the terminology is also not always consistently applied. For instance, Dewar
\citeyearpar{dewar2016,dewar2024} proposes a more subtle distinction
between ``de dicto'' and ``de re'' determinism, where his ``de dicto'' is clearly grade 1, but his ``de re'' determinism is weaker than Hawthorne's. (Actually, it is similar to Belot's D2 or Melia's ``Second Resolution'' \citeyear{melia}.)  \citet{four-man} use the terminology much like Dewar. \citet{cudek} tracks the same distinctions as \citet{four-man}, but using terminology that follows the model of \citet{butt-sub}.  And so on.}  For simplicity,
we will use Teitel's ``full determinism'' to gather together the
different senses of the highest grade of determinism.  We will argue
that previous attempts to define full determinism fail: it cannot be
cashed out either in terms of types of descriptions, nor in terms of
some equivalence relation on possible worlds. We do not conclude,
however, that qualitative determinism is the only kind of determinism
worth worrying about. In fact, it is possible to capture the
distinction between qualitative and full determinism, properly construed, in ``purely
formal'' terms, viz.\ in terms of uniqueness of extensions of
isomorphisms. Appearances to the contrary arise due to confusion about
(formal) interpretation and imprecision in the individuation of
theories.

For possible worlds $W,W'$, let us write $W\qe W'$ if $W$ and $W'$
match in qualitative description, in Hawthorne's sense (i.e.\ are
``qualitatively identical'' in Pooley's sense). Let's write $W\de W'$
if $W$ and $W'$ match in de re description, in Hawthorne's sense
(i.e.\ are ``intrinsically identical'' in Pooley's sense). While we
agree that it is possible to distinguish between grades of determinism,
we will show that such distinctions cannot be based on some
distinction between qualitative and de re equivalence of worlds. In
the remainder of this section, we will try to give the distinction
between $\qe$ and $\de$ a run for its money --- concluding that it
only gives intelligible answers for theories with names for all
objects.\footnote{The term ``name'' is loaded in the philosophy of language.  To forestall misunderstanding: when we write ``name'' we mean ``constant symbol in a formal language'', except where it is clear from context that we are entertaining other views.} In Section \ref{d3}, we propose a distinction between grades
of determinism that does not rely on a distinction between qualitative
and de re equivalence of worlds.

\subsection{Descriptions: Qualitative versus De Re}

Hawthorne and Teitel explain the distinction between $\qe$ and $\de$
in terms of a distinction between types of propositions.  Roughly
speaking, $M\qe M'$ means that $M\vDash \varphi$ iff
$M'\vDash\varphi$, for all qualitative propositions $\varphi$. And
$M\de M'$ means that $M\vDash\varphi$ iff $M'\vDash\varphi$, for all
de re propositions $\varphi$. Can we make rigorous sense of a
distinction between these two types of propositions?\footnote{Of
  course, we are hardly the first to ask this question. Indeed, there
  is a large literature on the distinction between qualitative and
  non-qualitative propositions \citep[see e.g.][]{kaplan, adams,
    rosenkrantz, dasgupta, cowling, hoffmann-kolss}. It seems to be
  widely agreed among metaphysicians that the distinction is
  intelligible and important, and yet that making it precise is
  elusive.  Several of the moves we make here will be familiar to some
  readers, though we do not attempt a general survey of the
  literature.  (For that, \citet{cowling} is especially helpful.)
  Instead, our goal is to show that it is unlikely that a sharp
  \emph{mathematical} distinction between qualitative and
  non-qualitative properties is available --- and thus, that the line
  of thought we are currently contemplating will not yield the sort of
  account we seek to defend.}

Hawthorne explains the distinction between two kinds of descriptions
as follows:
\begin{quote} The first --- the \emph{qualitative description} ---
  says everything that can be said about the intrinsic character of
  that history with one exception: it cannot name individuals or
  otherwise encode haecceitistic information about which particular
  individuals are caught up in that segment of world history. The
  second --- the \emph{de re description} --- includes the qualitative
  description and, in addition, all haecceitistic, singular
  information. \citep[p 239]{hawthorne} \end{quote} Granted, such a
distinction works fine in everyday life: in some contexts where we
have a name for an individual, we can give a specific description of
what properties that individual has, or we can give a general
description of something that has those properties. This is a
distinction between levels of generality.

And yet, for well known reasons, the distinction is not absolute.  For
instance, sometimes qualitative descriptions can contain singular
information. As Russell himself taught us, the sentence \begin{quote}
  (D) The current president of the US is bald. \end{quote} does not
contain the name of an individual, but is, in some sense, about a
particular individual. Similarly, Quine pointed out that a name such
as ``Pegasus'' could be replaced by a predicate ``pegasizes''
\citep{quine-on}.  Conversely, even names may not include all singular
information, at least in ordinary language:\footnote{In first-order
  logic, we regiment name usage to require names to refer uniquely.
  But that is a choice in how we set up our semantics, and it does not
  apply in ordinary language.}  ``James Weatherall'' may refer to a
philosopher of physics, or to his father, or his son, or to a former
All-American football player, or to a retired British Vice-Admiral.
In each of those cases, further descriptive information is needed to
disambiguate the reference of the name.

Perhaps more importantly, whether some description contains singular
information is apparently not a fact about the description itself, but
rather about the thing or things described. For example, ``Some
nobleman is bald'' could be about no individuals, or about one, or
about many.  And in any of these cases, it implies the less specific
proposition ``Somebody is bald''.  On these grounds, we are skeptical
about the idea that there is some significant distinction between
descriptions that encode information about individuals, and
descriptions that do not.

But let's try harder. Recall that a description is supposed to be
qualitative just in case it does not name any particular
individual. We will now consider several proposals for how to make
this idea precise. We assume that propositions are expressed, up to
logical equivalence, by sentences.  Thus, we assume that if sentences
are logically equivalent, then either they both name an individual, or
neither of them names an individual. If that were not the case, then
``$x$ names an individual'' would not be a property of the underlying
proposition, but only of some specific syntactic representation of
that proposition, and that would not help in defining a notion of
qualitative sameness of worlds.

\begin{ps} A sentence $\varphi$ names an individual just in case
  $\varphi$ contains a name. \end{ps} We have already seen arguments
against proposals of this sort, but stated in these terms it fails for
even simpler reasons, since it depends on superficial features of a
sentence that are not invariant under logical equivalence. Indeed,
$\varphi$ is logically equivalent to $\varphi \wedge (c=c)$, so by
this proposal, every sentence names an individual.

\begin{ps} A sentence $\varphi$ names an individual just in case there
  is a name $c$ such that for any sentence $\varphi '$, if $\varphi '$
  is logically equivalent to $\varphi$, then $\varphi '$ contains
  $c$. \end{ps}

This proposal would make it essentially impossible for sentences to
name individuals. Indeed, one could define a predicate
$\theta (x)\leftrightarrow (x=c)$, and then replace the de re sentence
$p(c)$ with the definitionally equivalent
$\exists !x(\theta (x)\wedge p(x))$.

\begin{ps} Given a background theory $T$, a predicate $\theta (x)$
  names some individual relative to $T$ just in case
  $T\vdash \exists !x\theta (x)$. \end{ps}

This condition is not strong enough. Even when
$T\vdash\exists !x\theta (x)$, there can be a model $M$ of $T$ such
that $M\vDash \theta (a)$, and another model $M'$ of $T$ such that
$M'\vDash \neg \theta (a)$. So, $T\vdash\exists !x\theta (x)$ does not
capture the idea that $\theta (x)$ names an individual.

We have tried, without any success, to give a formal (i.e.\
mathematical) account of the distinction between qualitative and de re
propositions. Perhaps our failure is unsurprising, since the
distinction between qualitative and de re descriptions was intended to
be a distinction in how those descriptions relate to the world, i.e.,
it is meant to be a distinction in what those descriptions are
\emph{about}, whereas we have been talking only about mathematical
objects (e.g.\ sentences of formal languages, models).

Here is another strategy.  Perhaps one could simply assume that there is a
relation $N(\varphi ,x)$ of ``naming'' that can hold for propositions
and concrete objects.

\begin{ps} A proposition $\varphi$ is \emph{qualitative} just in case
  $\neg\exists x N(\varphi ,x)$. \end{ps}

Such a proposal has all the advantages of theft over honest toil. But
to be more serious, we are still hung up on the practical concern of
whether the distinction would actually help us decide whether specific
scientific theories are deterministic, and saying that there is such a
relation $N$ does not tell us which sentences are qualitative and which
are de re. Nor would this proposal provide any concrete guidance on
the question of whether models are de re or qualitatively
equivalent. Metaphysicians would be within their rights to say ``there
is such a relation $N$'', but until they tell us how to detect that
$\exists xN(\varphi ,x)$ for a proposition $\varphi$, their proposal
is of no interest for understanding the metaphysical implications of
science.

Perhaps, you might think, we are straying too far from common
sense. Surely the following claim, which one might find in a science textbook, names a concrete physical object:
\begin{quote}
  (E) The perihelion of Mercury advances by approximately 43
  arcseconds per century more than what is predicted by Newtonian
  mechanics.
\end{quote}
But for reasons related to our arguments in section \ref{sec:interpretation}, we think this rests on a linguistic confusion.  Yes, the sentence involves a name in the colloquial sense.  But the sentence most certainly does not involve ``all singular, haecceitistic'' information about the planet closest to the sun.  It manages to convey the intended meaning only because of its integration with linguistic and scientific practice, our history and culture, established conventions, and so on.  Put another way, the term ``Mercury'' does not do anything, it is
people who do things. Names and sentences and propositions
do not name things, it is people who name them.  We may wish to include
an abstract intermediary, such as a proposition, that a person uses to
refer to something---and we might single out certain non-logical vocabulary as especially well-suited for playing a referential role in our languages. But without the person in the middle, the
proposition does not stand in an intrinsic relation to any concrete
thing.

We find that we have entered a wild and unfamiliar metaphysical
territory, far outside of our comfort zone. We have no carefully
worked out view of the metaphysics of propositions, and we do not
intend to take a stance on these issues. More power to them who
believe that they have a theory of propositions that will shed further
light on the elusive distinction between qualitative and full
equivalence of possible worlds.

What we do intend to take a stance on is that some definitions are more useful than others. For example, what does it mean to say that two sentences are synonymous?  As Quine pointed out, it is perfectly correct, and perfectly useless, to say that two sentences are synonymous if they express the same proposition. The problem is that nobody has ever figured out that two sentences are synonymous by comparing each of them with a proposition. In the same way, nobody will ever determine if a proposition $\varphi$ is de re by checking $\exists xN(\varphi ,x)$.


\subsection{Equivalence relations on models}

The attempt to distinguish qualitative from de re equivalence of
worlds via a distinction in kinds of descriptions strikes us as a dead end, at least from the perspective of trying to assess when a physical theory is deterministic.  Let's approach the issue from a different direction.  We are interested in two proposed senses of when two worlds ``agree''.  Perhaps we can capture this ``agreement of worlds''
more directly via models of a theory.

Hawthorne supplies a motivating example of worlds that are supposed to
be qualitatively equivalent but de re inequivalent.
\begin{quote} Consider a symmetrical world where there is a pair of
  qualitatively identical ships, one in each symmetrical half. Suppose
  the laws dictated that exactly one of the ships would sink, but left
  it undetermined which. Qualitative determinism might still hold of
  such a world, since the qualitative description of a world in which
  one ship sank need not depart in any way from the qualitative
  description of a world in which the other did. \citep[p
  243]{hawthorne} \end{quote} We share Hawthorne's intuition about
which counterfactual statements are true in this situation. In
particular, we think it is true that no matter which ship sinks, it
could have been the case that the other ship sank. And insofar as we
can specify the initial state of the world without specifying the
truth-value of future contingents, the initial state of the world does
not determine the final state. So there is clearly a sense in which
this example manifests a kind of indeterminism.

The key here is the phrase ``the other ship''. There are two ships to
begin with, and we can stipulate that one of them be called ``Lefty''
and the other ``Righty''. There is nothing in the description that
tells us which ship sinks --- it could be Lefty, or it could be
Righty. If Lefty does in fact sink, then it is still true that it
could have been Righty that sank. Indeed, there are \emph{two}
distinct possible final states of affairs, both of which are
compatible with the initial state of affairs.

All of this common sense talk can be shored up by formal model
theory. Consider a theory $T$ that says that there are exactly two
objects, two definite descriptions $L(x)$ and $R(x)$, and a single
predicate $P(x)$ to indicate which ship sinks at the latter
time. While a model $M$ of $T$ consists of a domain with two elements,
and extensions for $L,R$ and $P$, an initial segment of this model
simply omits the extension of $P$ --- i.e.\ it does not specify which
ship sinks. There is then a clear sense in which a single initial
segment can belong to two distinct models. For example, the initial
segment where $a$ is Lefty and $b$ is Righty belongs both to a model
where Lefty sinks, and to a model where Righty sinks. To be more
precise, a model where Lefty sinks is one where
$\exists x(L(x)\wedge P(x))$ is true, while a model where Righty sinks
is one where $\exists x(R(x)\wedge P(x))$ is true. Since models are
isomorphic only if they satisfy the same sentences, a model where
Lefty sinks is non-isomorphic to a model where Righty sinks.

There is also a clear sense in which any two models of $T$ are
qualitatively equivalent. The language of $T$ is the union of two
sublanguages $\{ P \}$ and $\{ L,R\}$, where the latter can be
considered as the ``haecceitistic'' vocabulary. In this case, we can
say that $M$ is qualitatively equivalent to $M'$ just in case $M|_P$
is isomorphic to $M'|_P$. It then follows that any two models of $T$
are qualitatively equivalent --- and this underwrites the intuition
that the process is qualitatively deterministic. We can then elevate
this idea into a proposal for how to determine when models of a theory
represent equivalent worlds.
\begin{ps} For models $M,M'$ of $T$, we say that $M$ and $M'$ are
  \emph{qualitatively equivalent} if $M$ and $M'$ are isomorphic qua
  structures for the qualitative sub-language. We say that $M$ and
  $M'$ are \emph{de re equivalent} if $M$ and $M'$ are isomorphic qua
  structures for the full language. \label{pooley} \end{ps} With this
distinction in place, we can then define:
\begin{fd}
  Let $T$ be $(\qs\cup \ds)$-theory. Then $T$ is deterministic just in
  case for any models $M,M'$ of $T$, and for any initial segments $U$
  of $M$ and $U'$ of $M'$, if $U$ and $U'$ are de re equivalent, then
  $M$ and $M'$ are de re equivalent. \end{fd} This definition gives
reasonable answers for the simple kinds of examples we have
considered. In particular, it explains the sense in which Hawthorne's
ship example fails to be fully deterministic.

Suppose now that we take Proposal \ref{pooley} and try to apply it to
a theory that purports to describe the actual world, e.g.\ General
Relativity. Here we face an obstacle: General Relativity does not
provide names for spacetime points. Can this be overcome?  In fact, something similar could be said of Hawthorne's theory. Hawthorne said ``consider a pair
of qualitatively identical ships'', and then \emph{we} proceeded to
name these ships ``Lefty'' and ``Righty''. We might say, then, that Hawthorne's original theory does not have names, in the technical sense of constant symbols assigned to domain elements, but we can extend that theory in a way that introduces names and breaks the symmetry between
the two ships. In doing so we introduce a new theory, which is such that every model of the original ``theory'' (i.e.\
``there are two ships and one sinks'') corresponds to two models of
the new name-enriched theory.

Imagine that you were learning General Relativity and the professor
said, ``a model of Einstein's Field Equations consists of a Lorentzian
manifold \dots '', and then you replied ``I'm going to add a name
`Bob', so that in each model, `Bob' refers to some particular
spacetime point.''  We have nothing against adding vocabulary to a
theory, but we insist that the result will typically be a
\emph{different} theory. In fact, GR+Bob has a predicate, `$x$ = Bob',
that is nomically isolated, i.e.\ it has no lawlike connections to any
other predicates or relations in the theory; and adding nomically
isolated predicates to a theory will necessarily result in an
indeterministic theory. For example, if N is the theory of inertial
particle motion in Newtonian absolute spacetime, then adding two
one-place predicates $A$ and $B$ will result in an indeterministic
theory --- since Newtonian mechanics does not say whether a particle
that satisfies $A$ at one time will satisfy $B$ at some later time.

Now, one might think that the issue in theses examples is not really
about adding names to a theory, because somehow the resources needed
to talk about individuals are already available without any
extension. Indeed, in typical discussions of the hole argument in
general relativity, the labels for spacetime points are not taken to
be part of the syntax of the theory, but are assumed to be added by
us, after we choose a specific spacetime $M$. i.e.\ we don't assume
that there is some name `Bob' that refers to a point in every
spacetime. Instead, once we are given a spacetime $M$, we note that it
is a non-empty set, and we say ``consider some $p\in M$''. We then
note that $p$ has some property $\varphi$ in $M$, but there is a
mathematical construction that generates another spacetime $M'$ based
on the same set, and in which $p$ does not have property
$\varphi$. The thought then is that $M$ and $M'$ represent de re
inequivalent worlds.

We could look at Hawthorne's example in the same fashion. Don't
introduce names `Lefty' and `Righty' at the start, but just write down
a model $M=\langle \{ a,b\},a\rangle$ for ``the theory of two ships,
where one sinks''. Clearly $M'=\langle \{ a,b\},b\rangle$ is also a
model for this theory, and we could imitate the hole argument by
constructing $M'$ from $M$. The fact that $M\neq M'$ as sets seems
then to underwrite the claim that there are two de re inequivalent
worlds.

But what is the basis for saying that $M$ and $M'$ represent de re
inequivalent worlds?  Consider a possible world with a pair of
objects, one of which is $P$. We can name the thing that is $P$ in
various ways: we could name it $a$, or we could name it $b$. So why
not take $\langle \{ a,b\},a\rangle$ and $\langle \{b,a\} ,b\rangle$
to be two different representations of the exact same world? Nor will
it help here to say: ``but $a$ and $b$ are not names of things, they
are the things themselves''. That is just a bit of nonsense, for
reasons already discussed.\footnote{Recall fn. \ref{urelement} and the surrounding discussion.}

The claim that $M$ and $M'$ represent de re inequivalent possibilities
is less solid than it might initially seem. How might it be justified?
The obvious answer here is that $M$ and $M'$ are set-theoretically
distinct: $M$ assigns $\{ a\}$ to $P$, while $M'$ assigns $\{ b\}$ to
$P$. This suggests the following proposal:
\begin{ps} Models $M$ and $M'$ represent de re equivalent worlds only
  if $M=M'$ in the sense of Zermelo-Fraenkel set
  theory. \label{id} \end{ps} We ourselves are inclined to think that
set-theoretic identity of models is of no significance here; we think
that isomorphism of models is the only significant standard of
sameness.\footnote{Reasons for thinking this in a related context are
  given by \cite{regarding} and \cite{bradley}; other reasons, related
  to the arguments in section \ref{sec:interpretation}, are discussed
  below.  Ultimately, the issue is that insisting on set-theoretic
  identity of models involves layering interpretation on
  interpretation in a way that fundamentally confuses what the theory
  expresses.  Perhaps ironically, given the subsequent literature,
  while \citet{kaplan} both introduces ``haecceitism'' to contemporary
  philosophy and also suggests, in the somewhat different context of
  Kripkean modal semantics, that set theoretic identity is a
  reasonable way to represent haecceitistic facts, he \emph{also}
  suggests that the anti-haecceitistic philosopher has an easy way to
  avoid taking set theoretic identity to have this significance,
  viz. by introducing suitable isomorphisms.  All of the theories we
  consider have such isomorphisms defined.  We also note that
  \citet[p. 103]{pooleyOld} makes a closely related observation, and
  also cites Kaplan as recognizing that even when these isomorphisms
  are introduced, the sets we use to build models ``will be open to a
  haecceitistic misinterpretation''.}  But we do not need to convince
you of that stronger position.  It is enough to point out that
Proposal \ref{id} supplies far too many possible worlds, viz.\ as many
as there are sets. For example, when Hawthorne begins describing his
example by saying ``consider two ships'', we are left wondering
whether the first ship is $a$ or $\{ a\}$ or $\{ \{ a\} \}$ or \ldots
We cannot speak for metaphysicians, but no physicist would say that it
matters which element $a$ we use to reprsent the first ship, as long
as we use a different element $b$ to represent the second ship. In
fact, the reason we write ``$a$'' and ``$b$'' as opposed to something
like ``$\{ \emptyset \}$'' and ``$\{ \emptyset ,\{ \emptyset \}\}$''
is because we have no good reason to choose one set over another to
stand in for our ships. We are not really supposed to consider the
first ship to be a \emph{specific} object in a model of ZF set theory,
but as a generic element of a two-point set. More to the point,
standard practice in mathematical physics allows that the same
possibility be represented by distinct sets. Nobody ever asks whether
a spacetime $M$ contains $\{ \emptyset \}$ --- because that would be a
representational artifact.

Besides, even if a spacetime model really
is a set, we should be skeptical that that should make a difference to whether general relativity is deterministic.  After all, the working physicist never cares which set it is, as long
as it has the right structure, and this structure is specified by the
syntax of the scientific theory --- not by set theory. Therefore,
set-theoretic equality of models should not be viewed as a necessary condition for de
re equivalence of the corresponding worlds, since set theoretic equality is irrelevant to any sort of physical equivalence. Proposal \ref{id} only
sounds promising until one thinks about how many different sets there
are, and about how irrelevant it is which set we use.

We've been trying to figure out when models of a theory represent de
re equivalent possibilities. The first proposal along these lines, proposal \ref{pooley}, was that we add names
up front (enriching the syntax of the theory), and declare that models
represent de re equivalent possibilities only if they are isomorphic
relative to the name-enriched language. But that proposal won't help
determine whether the original theory is deterministic, since adding
names to a theory results in a new theory. The second proposal, proposal \ref{id}
attempts to cash out de re equivalence in terms of set-theoretic
identity of models. Like the first proposal, this second proposal
tacitly replaces the original theory, with its criterion for identity
of models (viz.\ isomorphism), with a new theory that recognizes more
properties of models and fewer isomorphisms between them. For example,
the Special Theory of Relativity says that spacetime is a
four-dimensional affine space with a Lorentzian metric. It does not
recognize any difference between two such affine spaces, even if there
are set-theoretic differences --- e.g.\ the first contains
$\{ \emptyset \}$ and the second does not. Now, one could adopt a new
theory in which each distinct set of cardinality $2^{\aleph _0}$
represents a physically inequivalent possible world. This new theory
is indeterministic because there are two spacetimes $M$ and $M'$ that
are identical up to $t=0$, but where $\{ \emptyset \}\in M$ while
$\{\emptyset \}\not\in M'$. We will not argue that this new theory is
incoherent, but we insist that it is not the Special Theory of
Relativity.

We've been at pains to argue that the distinction between
``qualitative'' and ``de re'' equivalence is too opaque to be of much
use in deciding whether a theory is deterministic. That being said, we do agree with
Hawthorne, Pooley, Teitel, and others that Lewis' criterion for
determinism, combined with ``isomorphism is equality'', fails to
detect some failures of determinism. Where we differ is that on our view, the problem is not that we
need a binary relation that is finer-grained than isomorphism. Instead, it is
that we need to keep track of the functions, i.e.\ the isomorphisms,
that ground the claim that two models are isomorphic. While previous
discussions have analyzed determinism in terms of binary relations
between worlds, a more illuminating analysis can be carried out in
terms of functions between worlds --- what Butterfield and Belot call
``duplications''.

\section{Lewis formalized} \label{d3}

Distinguishing senses of determinism by distinguishing types of
descriptions or equivalence classes of models did not get us very far.
But in fact, we think something very much in the spirit of the
distinction that Hawthorne, Teitel, and others have aimed at is
available, and that an adequate definition of determinism that
captures something like the intuition behind Full Determinism,
properly understood, is available.  To see it, though, requires us to
shift back to the other post-Lewisian thread, the one that was
prematurely cut off by \citet{belot-do,belot}. We intend to take this
line of development back up, and to argue that it is a genuine
problem-solver for questions about determinism. In other words,
insofar as Lewis was continuing the Carnapian program of explication,
then he was on the right track.

This line of development consists of three positive papers, and two
negative papers. On the positive side, Butterfield
\citeyearpar{butt-sub,butt-ein,butterfield} points out that there is
an imprecision in Lewis' talk of ``diverging worlds''. Since Lewis'
counterpart theory entails that no individual can exist in two worlds,
distinct worlds can never really have overlapping initial
segments. Butterfield then proposes that Lewis needs a notion of a
\emph{duplication map} $g:U\to U'$, where $U$ is an initial segment of
$W$, and $U'$ is an initial segment of $W'$. With this notion in hand,
Butterfield explains that there are two precisifications of Lewis'
notion of determinism, a stronger one ``DM2'', and a weaker one
``DM1''. Roughly speaking, DM2 says that if there is a duplication
$g:U\to U'$ of initial segments of worlds $W,W'$, then there is a
duplication $f:W\to W'$ of these worlds.\footnote{This discussion
  occurred in the context of the hole argument, and so the definitions
  were originally stated in terms of manifolds, tensors, and smooth
  mappings.  But the structure of Butterfield's definitions are
  independent of these details.  For further development in those
  terms, specifically for GR, see \cite{four-man}.} Then Butterfield
goes on to apply this analysis to the hole argument, and shows that GR
satisfies DM2.

This is an amazing outcome. Butterfield seems to have blocked the hole
argument, vindicating substantivalism and Lewisian counterpart theory,
not to speak of formal definitions of determinism. But then along came
\citet{belot}, who showed that even Butterfield's stronger condition
DM2 was too liberal by giving an example of a clearly indeterministic
process that satisfies it. He then went on to give two refined and
strengthened versions of DM2, but immediately provided counterexamples
to them. The upshot seems clear: do not try to turn Lewis'
metaphysical definition of determinism into a Carnapian explication,
because formal definitions will never capture the full sense of
determinism. At least that seems to have been the lesson that many
philosophers --- Belot included --- took away from his arguments.

We have already argued against this general posture.  But we also
think Belot's arguments fail on their own terms.  That is, we think
that Belot himself gave a promising formal definition of
determinism. We will now argue that his ``counterexample'' is nothing
of the sort. It does \emph{not} show the inadequacy of his
precisification of Lewis' diverging-worlds definition of determinism.

Belot's (1995b) first definition of determinism is essentially a
direct transcription of Butterfield's DM2:
\begin{quote} \textbf{D1:} A world $W$ is deterministic if, whenever
  $W'$ is physically possible with respect to $W$ and $t,t'$, and
  $f:W_t\to W_{t'}$ are such that $f$ is a duplication, there is some
  duplication $g:W\to W'$. \end{quote} Belot argues that D1 does not
capture determinism in its fullest sense, since there are
indeterministic processes that are D1-deterministic. Belot presents an
example of a centrally loaded column that must buckle in some
direction --- but which direction is undetermined. Since any two final
states are related by a duplication map, this example satisfies D1
since. However, the example still seems to be indeterministic in some
sense.

After dismissing D1, Belot considers the following strengthened
criterion:
\begin{quote}
  \textbf{D2:} $W$ is deterministic if, whenever $W'$ is physically
  possible with respect to $W$, and $t,t'$, and $f:W_t\to W'_{t'}$ are
  such that $f$ is a duplication, there is some duplication
  $g:W\to W'$ whose restriction to $W_t$ is $f$.
\end{quote}
The key difference between D1 and D2 is that the latter requires a
relationship between the duplication $g:W\to W'$ and the duplication
$f:W_t\to W_{t'}$, viz.\ $f$ is the restriction of $g$ to $W_t$. The
fact that D2 is genuinely stronger than D1 depends on the assumption
that ``agreement'' can be witnessed by various functions. Indeed, if
agreement were a binary relation on worlds (or world segments), then
D1 would imply D2. This might explain why modal metaphysicians have
overlooked D2 (or the even stronger version D3, that we will soon
consider). Modal metaphysicians have tended to think in terms of
binary relations on worlds, whereas D2 asks us to keep track of
different ways that worlds can be matched with each other.

Belot then provides a counterexample to D2: a single $\alpha$ particle
decays into two identical $\beta$ particles. This example satisfies
D2: if $f$ is a duplication of initial segments, then $f$ can be
extended to a duplication of worlds. In fact, since the $\beta$
particles are assumed to be identical (i.e.\ related by a symmetry),
$f$ can be extended in two ways. But this very non-uniqueness of
extensions makes Belot (and us) think that the process is not actually
deterministic. If $\beta _1$ and $\beta _2$ are the particles in one
world, then we share Belot's intuition that the following
counterfactual is true:
\begin{quote}
  (*) $\beta _1$ could have been $\beta _2$. \end{quote} (Note,
however, that the truth of (*) does not depend on the existence of a
distinct world.)

We agree, then, that D2 does not capture determinism in the strongest
sense. But Belot has yet another proposal.\footnote{Belot's D3 seems
  to have been independently rediscovered by \citet{landsman} and
  \citet{cudek}.}
\begin{quote}
  \textbf{D3:} A world $W$ is deterministic if, whenever $W'$ is
  physically possible with respect to $W$, and $t,t',W'$ and
  $f:W_t\to W'_{t'}$ are such that $f$ is duplication, then there is
  exactly one duplication $g:W\to W'$ which extends~$f$.
\end{quote}
That is, D3 strengthens D2 by requiring uniqueness of extension of a
duplication map of initial segments.

Some of the details of D3 are inessential, and it can easily be made
into a schematic that applies to just about any scientific theory. For
example, while D3 is formulated in terms of possible worlds, we will
sometimes talk instead about models (of a theory). Similarly, D3 takes
the determiner to be a time-slice $W_t$, but we can take it to be
other parts of a world or a model, e.g.\ an initial segment of a
possible world \citep[see][]{lewis-work,butterfield}, or an initial
data surface embedded in a four-dimensional manifold
\citep[see][]{landsman}. The details may differ, but all of these
cases conform to the following schematic:\footnote{As we note later,
  this schematic looks suggestively like a version of implicit
  definability.}
\[ \begin{tikzcd}
    M \arrow[r, "g", dashed] & M' \\
    U \arrow[u, "i"] \arrow[r, "f"'] & U' \arrow[u,
    "i'"'] \end{tikzcd} \] Here $i:U\to M$ and $i':U'\to M'$ are the
embeddings of initial segments into the entire history, and
$f:U\to U'$ is an isomorphism of initial segments. D3 then says:
determinism holds just in case any isomorphism of initial segments
extends uniquely to an isomorphism of worlds.

In what follows, we will argue that D3 captures the fullest sense of
determinism that is reasonable to apply to a physical theory. For now,
we will simply point out that D3 holds for General Relativity. To be
more precise, D3 holds for four-dimensional, maximal, globally
hyperbolic, vacuum solutions to Einstein's equation. This result
follows from the conjunction of the Choquet-Bruhat-Geroch theorem and
Geroch's \citeyearpar{geroch} rigidity theorem --- see the accounts in
\citep{landsman} and \citep{dialogue}.\footnote{A more careful
  treatment of determinism in GR is given in \cite{four-man}. Our
  attitude here should be clarified by what we say in section
  \ref{sec:interpretation} about electromagnetism. The difference
  between two versions of EM manifests itself in different choices of
  isomorphisms between models; and one of these two choices leads to a
  better theory (in our opinion, and in the opinion of most
  physicists).}

\section{Belot against D3}

We think that D3 is an excellent definition of when a theory is
deterministic. But that is not the conclusion that its architect,
Belot, drew. In fact, Belot's criticism of D3 made it all but
invisible to philosophers for thirty years, until it reappeared in
work by \citet{close,landsman,cudek}. In this section, we consider
Belot's purported counterexample to D3, and we show that it
underspecifies the relationship between spacetime and its material
contents.
\begin{quote}
  In this example, $W$ is a world with spacetime points and Newtonian
  spacetime structure. It initially contains a single $\alpha$
  particle. The laws of nature decree that one year later, at $t=1$,
  the $\alpha$ particle decays into continuum many $\beta$ particles;
  arranged so that at time $t$, the $\beta$ particles form a spherical
  shell of radius $t$; with each $\beta$ particle moving away from the
  center of the sphere along its radius.\footnote{We changed the time
    scale for ease of exposition.} \citep[p 193]{belot} \end{quote}
Let's follow Belot's own presentation of the conflicting claims that
this world is deterministic according to D3, but is actually
indeterministic.

Belot's argument that $W$ satisfies D3 relies on the following claims:
\begin{enumerate}
\item Duplications must preserve metric relations between spacetime
  points. Hence, a duplication map from $W$ to itself must be a
  symmetry of Newtonian spacetime.
\item Duplications must preserve the relation of ``$x$ is located at
  $y$'' that holds between a material object and a spacetime point.
\item The only Newtonian symmetry that preserves the worldline of the
  $\alpha$ particle is a rotation around the timelike line that
  extends that particle's trajectory.
\end{enumerate}
Supposing that these three claims hold, then if $f$ is a duplication
of $W$'s initial segment, then $f$ is a rotation. It follows that $f$
has a unique extension to all of the spacetime points of $W$; and
since location relations must be preserved, $f$ has a unique extension
also to material objects. Therefore the extension of $f$ is unique,
and D3 is satisfied.

Belot then argues that the example should be conceived of as
indeterministic. His argument relies on the following claim:
\begin{quote}
  There is a legitimate counterpart relation $g'$ (not a duplication)
  that is the identity on spacetime points, but not the identity on
  material objects. That is, $g'$ does not preserve the relation ``$x$
  is located at $y$.''\footnote{``$g'$ is a counterpart relation that
    tells us that $\beta _1$ could have been $\beta _2$ etc. $g$ and
    $g'$ license us to say that there are two possible futures for
    $W_t$.''}
\end{quote}
We agree completely that if $g'$ is ``a legitimate counterpart
relation'', then the correct judgment is that the example is
indeterministic. In fact, the existence of two counterpart relations
$g$ and $g'$, both of which extend $f$, would indicate a failure of
D3. The problem is not with D3 as a criterion, but with an ambiguity
about which functions from $W$ to itself are relevant for deciding
whether $W$ is deterministic. That is, whether or not a world is
deterministic depends crucially on which functions from that world to
itself are actually duplications. Nor do we think that this question
--- what are duplications? --- can be settled by apriori
metaphysics. Formulating a scientific theory includes a specification
of which maps between models are isomorphisms (i.e.\ duplications),
and that specification is decisive for whether that theory is
deterministic.

We have seen that there are two ways to read Belot's example: either
the relation between $\beta$ particles and spacetime points is
nomologically necessary, or it is not. If this relation is
nomologically necessary, then the symmetries of a spacetime are
uniquely paired with symmetries of its material contents. If this
relation is not nomologically necessary, then there could be a
spacetime symmetry that is not uniquely paired with a symmetry of its
material contents. A theory that describes the former situation can be
expected to satisfy D3, while a theory that describes the latter
situation can be expected not to satisfy D3. In this section, we will
give simple examples of such theories.

Newtonian spacetime has the feature that spatial points maintain their
identity over time, and so it makes sense to talk about whether an
object is changing its position over time. (This in contrast to
Galilean spacetime.) Newtonian spacetime also has a rather small group
of symmetries: uniform shifts and rotations around timelike (vertical)
lines. These features of Newtonian spacetime play an important role in
the setup of Belot's example, but we can capture the essential points
with a simple toy model. It will suffice to have three locations
(left, center, and right), and two times ($t=0$ and a final time
$t=1$). We won't need shifts (which play no role in Belot's example),
but we allow for rotations around the center location, i.e.\
permutations of right and left.

There are two ways to set up such a framework, which are equivalent
for Newtonian spacetime, and only come apart for more sophisticated
examples.
\begin{enumerate}
\item Represent spacetime by a family of types $S_t$, with $t$ a time
  parameter, and postulate a ``persistence'' relation (the analogue of
  an affine connection) between the types.\footnote{Here we use
    ``type'' and ``sort'' synonymously, both in the sense of
    many-sorted logic \citep[see][]{lps}. We are using this framework
    for its flexibility, and without any commitment to type theory as
    a foundation of mathematics.} The persistence relation can be
  represented by isomorphisms with compatibility relations. To
  represent Newtonian spacetime, we assume a unique isomorphism
  $\delta _{t,t'}:S_t\to S_{t'}$.
\item Represent space by a type $S$ and time by another type $R$, so
  that spacetime is represented by the product type $R\times
  S$. \end{enumerate} The advantage of the first, more complicated,
setup is that it generalizes more easily to other spacetime theories,
e.g.\ Galilean spacetime. The advantage of the second setup is that we
don't have to keep track of sorts. Since Belot's example assumes
Newtonian spacetime, we will start with the second approach.

\begin{figure}
\begin{tikzpicture}
  \draw[thick] (0,0) rectangle (6,4); 
  \draw[thick] (0,2) -- (6,2); 
  \draw[thick] (2,0) -- (2,4); 
  \draw[thick] (4,0) -- (4,4); 

  \node[font=\Large\bfseries] at (3,1) {$\alpha$};
  \node[font=\Large\bfseries] at (1,3)
  {$\beta_1$}; 
  \node[font=\Large\bfseries] at (5,3) {$\beta_2$}; 
  \end{tikzpicture}
  \caption{Space with three places and two times, with two $\beta$
    particles at the later time. Since symmetries are assumed to fix
    $\alpha$, the center blocks can be omitted from the model without
    changing the conclusions we draw.}
\end{figure}

Let $\Sigma$ be a signature with a sort symbol $S$ for spatial
points. We could then add the axiom that there are three things of
sort $S$, corresponding to the three particle positions. But it is
simpler just to ignore that $\alpha$ particle, which effectively
defines a constant symbol (a name for its location). Thus we take as
our first axiom:
\begin{quote}
There are two things of sort $S$.
\end{quote}
We now let $B$ be a sort symbol for the $\beta$ particles, and we add
a second axiom:
\begin{quote}
There are two things of sort $B$.
\end{quote}
The crucial decision now is whether to add an additional relation
symbol that relates $\beta$ particles to spatial points. To do so
would amount to positing a lawlike relation between material objects
and spatial points --- a relation that must be preseved by all
symmetries of models and all isomorphisms between models.

Suppose first that we do \emph{not} add any such relation, and let $T$
be the bare theory of spatial points and $\beta$ particles. It is clear
that $T$ does not satisfy D3: any symmetry $g$ of spatial points is
compatible with two distinct symmetries of $\beta$ particles (the
identity and the flip).

Suppose now that $L$ is a relation between $\beta$ particles and
spacetime points, and let $T_q$ be the theory that says that each
$\beta$ particle is located at a unique spatial point at $t=1$. A
model $M$ of $T_q$ consists of two two-element sets $S^M$ and $B^M$,
and a one-to-one correspondence $L^M$ between $B^M$ and $S^M$. It is
easy to see that the models of $T_q$ satisfy D3. Indeed, let
$i:S\to M$ and $i':S'\to M'$ be embeddings of initial conditions. Here
$S$ and $S'$ are sets with two elements, and $M$ (respectively $M'$)
adds a second set $B$ (respectively $B'$) and an isomorphism
$L:B\to S$ (respectively $L':B'\to S'$). Any isomorphism $f:S\to S'$
of space is compatible with precisely one isomorphism
$L'\circ f\circ L$ of $\beta$ particles:
\[ \begin{tikzcd}
  B  \arrow[r, dashed]     & B' \\
  S\arrow[u, "L"] \arrow[r, "f"'] & S' \arrow[u, "L'"']
\end{tikzcd} \] Therefore, $T_q$ satisfies D3.

The theories $T$ and $T_q$ illustrate the obvious fact that whether or
not a world $W$ is deterministic depends on the relations that world
bears to other worlds. If all nomologically possible worlds have the
same pairing of $\beta$ particles and spacetime points, then the
process of beta decay is deterministic. If some nomologically possible
worlds have different pairings of $\beta$ particles and spacetime
points, then the process of beta decay is indeterministic. This
difference in verdicts is not a bug in definition D3; it is a feature
of D3 that it depends on a precise specification of permitted
duplications.

There is yet another theory that could be said to capture Belot's
example. Let $T_h$ be a theory with a sort symbol $S$ for spatial
points, and with a name $b$ for one of the two $\beta$ particles. We
consider $\ulcorner x=b\urcorner$ to represent the property that the
spatial point $x$ is occupied by $b$ at $t=1$. In this case, a model
$M$ of $T_h$ consists of a set $S^M$ with two elements, and a
distinguished element $b^M\in S^M$. If $M$ and $M'$ are models of
$T_h$, then an isomorphism $g:M\to M'$ consists of a function from $S$
to $S'$ such that $g(b^M)=b^{M'}$. It follows that there is a unique
isomorphism between any two models of $T_h$.

For $T_h$, an initial condition is a set $S$ with two elements, while
a final condition is the same set $S$ plus the choice of one of the
two elements $b^M\in S$. This choice leads to a reduction of symmetry,
and so to a violation of D3. Indeed, consider the isomorphism
$f=1_S:S\to S$ of initial conditions. Let $b^M\in S$, and let $b^{M'}$
be the other element of $S$. Then there is no isomorphism $g:M\to M'$
that completes the following diagram:
\[ \begin{tikzcd}
    M \arrow[r, dashed, "g"] & M' \\
    S \arrow[u, "i"] \arrow[r, "f"'] & S \arrow[u,
    "i'"'] \end{tikzcd} \] Therefore, $T_h$ does not satisfy D3.

Let's take stock. The theories $T,T_q$ and $T_h$ describe similar
worlds: at the initial time there are two locations and no material
particles; and at a subsequent time, each location is occupied by a
distinct $\beta$ particle. But there is a subtle difference in how
these theories describe the world: $T_q$ fixes the relationship
between $\beta$ particles and spatial points, and its symmetries
preserve this relationship. In contrast, $T$ and $T_h$ permit
different matchings between $\beta$ particles and spatial points, and
the symmetries of spatial points are not uniquely paired with
symmetries of $\beta$ particles.

Belot's clever example underspecifies the relationship between
material particles and spatial points. According to one specification,
material particles are nomologically tied to particular spatial
points. In this case, the world $W$ is unambiguously
deterministic. According to another specification, one and the same
material particle can be located at different spatial points; and, in
this case, a symmetry of space does not extend uniquely to a symmetry
of its material contents. In this latter case, the world $W$ should be
judged to be indeterministic. The lesson is that determinism depends
on which relations are preserved by duplications, and theories tell us
what those relations are.

We conclude this section on an ironical note: Belot argues that
``determinism is not a formal property of uninterpreted theories.''
But the example he describes cannot be judged either to be
deterministic or indeterministic until we are told precisely which
mappings between worlds are duplications, and this latter
specification is a formal property of a theory.

\section{Bridging a non-existing gap}

\citet{teitel} argues that metaphysicians have an important job in
uncovering what modal-metaphysical commitments might be required to
maintain the consistency of spacetime substantivalism with full
determinism. He poses the challenge as ``bridging the gap'' between
qualitative and full determinism.\footnote{In more recent work, \citet{TeitelSubstantivalist} has backed away from this challenge.  Instead, he bites the bullet and accepts that the haecceitist substantivalism that he prefers is indeterministic, properly construed.  But he also maintains that any viable metaphysic that can make sense of the modal requirements imposed by ``sophisticated'' substantivalists suffers from its own challenges related to determinism, viz., that they make determinism too ``cheap'' by denying the metaphysical possibility that given spacetime points could have had different properties than they do.  There is much to say in response to this rich paper, but for present purposes his earlier work is a better foil to make our point---even if he no longer defends the views we cite.}
\begin{quote} We need a doctrine that \dots bridges the crucial gap
  between GR's qualitative determinism and its full determinism
  (thereby resolving both the original hole argument and my revised
  hole argument). \citep[p 379]{teitel} \end{quote}
\begin{quote} Any of those three anti-haecceitistic doctrines suffices
  to bridge the gap between GR's qualitative and full
  determinism. \citep[pp 359-360]{teitel} \end{quote}
\begin{quote} I deliberately set up the issues surrounding the hole
  argument by directly discussing modality and which doctrines imply
  the right modal correlations to bridge GR's qualitative and full
  determinism, rather than following the standard practice in the
  literature of theorizing primarily in terms of mathematical solution
  spaces and what we use them to represent. \citep[p
  388]{teitel} \end{quote} We agree that if there were a gap between
the sense in which GR is deterministic and some more metaphysically
significant kind of determinism, then it would be worth inquiring into
what metaphysical commitments are needed to bridge this gap. We claim,
however, that there is no such gap.\footnote{Of course, this does not
  mean we do not recognize different senses, or ``strengths'', of
  determinism.  For instance, Belot's D1 and D2 are weaker than our
  preferred D3, and may well be viewed as capturing senses of
  ``qualitative determinism''.  The crucial point is that GR is
  deterministic in a stronger sense than either of these, and so there
  is no gap for determinism in GR.}

The hole argument raises many technical issues that are beyond the
scope of this paper. Fortunately, the literature of the past thirty
years offers numerous toy examples that are supposed to be analogous
to GR in being qualitatively, but not fully, deterministic.  (See
Figure \ref{fig2}). We have encountered two already: Belot's $\beta$
particles and Hawthorne's ships. The doubly-symmetric world described
by \citet{melia} provides yet another.  We will now show how Hawthorne
and Melia's examples, like Belot's, illustrate the distinction between
D1 and D3.

\begin{figure}[h]
\begin{tikzpicture}
    \draw[thick] (0,0) circle (3cm);
    \node at (0,3.5) {\textbf{qualitative determinism}};

    \draw[thick] (0,0) circle (1.5cm);
    \node[align=center] at (0,0.2) {\textbf{full} \\ \textbf{determinism}};

    \node (GR) at (-1.8,1.2) {?};
    \node (Ships) at (1.8,1.2) {\textbullet};
    \node (Belot) at (-1.6,-1.4) {\textbullet};
    \node (Melia) at (1.6,-1.4) {\textbullet};

    \node[left] (GR_label) at (-4,2) {GR};
    \draw[->, dashed] (GR_label) -- (GR);

    \node[right] (Ships_label) at (3.5,1.2) {Hawthorne's ships};
    \draw[->] (Ships_label) -- (Ships);

    \node[left] (Belot_label) at (-4,-1.5) {Belot's $\beta$ particles};
    \draw[->] (Belot_label) -- (Belot);

    \node[right] (Melia_label) at (3,-1.5) {Melia's symmetric world};
    \draw[->] (Melia_label) -- (Melia);

  \end{tikzpicture}
  \label{fig2}
  \caption{Theories that are supposed to be qualitatively, but not
    fully, deterministic.}
\end{figure}
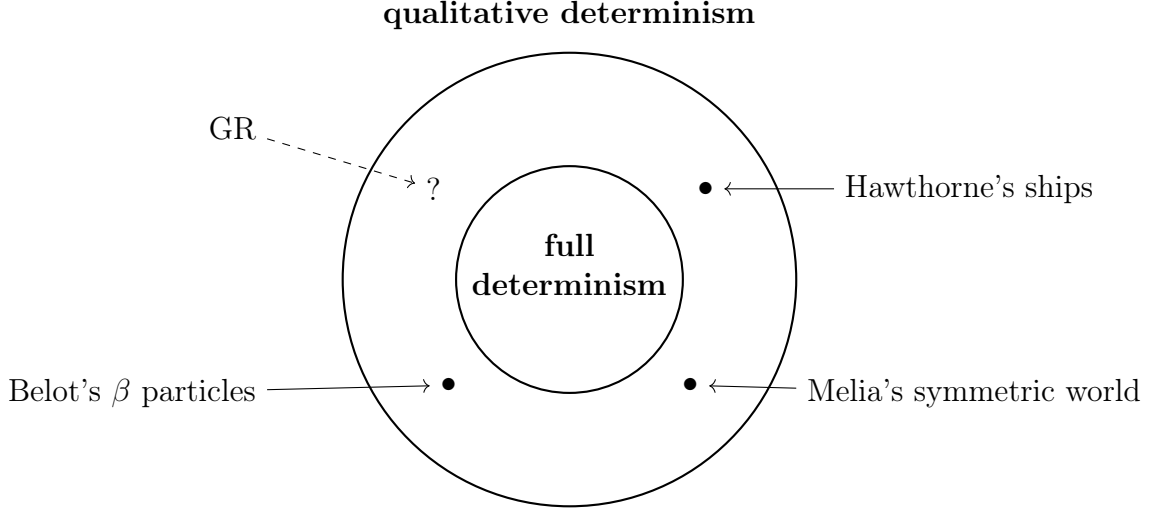

As a warmup, consider one of Melia's simpler (but more entertaining)
examples:
\begin{quote}
  We could imagine a collection of bald philosophers, sitting in a
  circle. It is a law that one of them will grow a single hair. But,
  by the symmetry of the situation, \emph{any} of the philosophers
  could be the lucky one. Again, our intuition is that there are many
  qualitatively isomorphic but distinct possibilities, each
  representing a different way in which the situation could
  evolve. \citep[p 650]{melia} \end{quote} Intuitively this example is
qualitatively deterministic, since any two possible final conditions
are qualitatively identical. And yet, this example clearly isn't fully
deterministic, since the law does not stipulate \emph{which}
philosopher will grow a hair. We agree that this example is
deterministic in one sense, but not another. Only we think that the
correct analysis is that this example is D1 but not-D3.

To see this, suppose that there are initially $n>1$ philosophers, and
that $p_t$ represents the property of being bald at time $t$. Let $T$
be the theory with axioms $\forall x p_0(x)$ and
$\exists !x\neg p_1(x)$. Since $T$ entails that $p_0$ holds of all
things, the predicate symbol $p_0$ plays no role in the analysis, and
we may drop it. If we set $p=p_1$ for notational simplicity, then a
model $M$ of $T$ is determined by a set $S$ and a singleton subset
$p^M\subseteq S$. If $M,M'$ are models of $T$ with the same initial
conditions (i.e.\ the same domain $S$), then there is at least one
bijection $g:S\to S'$ such that $g(p^{M})=p^{M'}$. Thus, $g:M\to M'$
is an isomorphism, and D1 is automatically satisfied.

We now show that D3 fails. Let $M$ be a model of $T$ whose domain $S$
has two elements; and let $M'$ be the model that has the same domain
as $M$, but where the extension $p^M$ has been switched to the other
element of the domain, i.e.\ $p^M\neq p^{M'}$. (Note that $M$ and $M'$
are isomorphic models.)  Then the identity $1_S$ is an isomorphism
between the initial conditions of $M$ and $M'$. However, if $1_S$ were
an isomorphism of $M$ to $M'$, then it would follow that $p^M=p^{M'}$,
contradicting the definition of $M'$. Therefore, $T$ is
D3-indeterministic.

Melia's bald philosophers example is supposed to be a paradigm case
where qualitative, but not full, determinism holds. But our split
intuitions about this example can be explained by a more clear
distinction, viz.\ that between D1 and D3. The bald philosophers
example does not provide any support for the legend that there is a
deeper, metaphysical sense of determinism that cannot be captured by a
formal definition.

But Melia has another trick up his sleeve: an example so clever in
conception that one feels sure that the quest for a formal definition
of determinism will have to be abandoned.
\begin{quote}
  Consider a world whose initial conditions consist of the following
  situation (see Figure \ref{joe}). The two white particles are
  duplicates of each other and the two black particles are duplicates
  of each other. The laws in this world dictate that, after a certain
  fixed period of time, each black particle will start moving at a
  fixed velocity in a straight line towards a white particle, and that
  the two black particles will move towards \emph{different} white
  particles. Using names for the objects found in the situation above,
  after a fixed amount of time either \textbf{c} will head towards
  \textbf{b} and \textbf{d} will head towards \textbf{a}, or
  \textbf{c} will head towards \textbf{a} and \textbf{d} will head
  towards \textbf{b}. \citep[p 661]{melia}
\end{quote}
Once again, it is clear that the two possible final conditions are
qualitatively identical, and hence that this example is qualitatively
deterministic. But surely, one thinks, there is a \emph{haecceitistic}
difference between the two final conditions, and so the example is not
fully deterministic.

\begin{figure}[h]
  \centering
\begin{tikzpicture}

\node[draw, circle, minimum size=1cm] (a) at (-2.5, 0) {};
\node[draw, circle, minimum size=1cm] (b) at (2.5, 0) {};

\node[draw, circle, fill=black] (c) at (0, 1) {};
\node[draw, circle, fill=black] (d) at (0, -1) {};

\node[above, yshift=0.2cm] at (a.north) {\textbf{a}};
\node[above, yshift=0.2cm] at (b.north) {\textbf{b}};
\node[above, yshift=0.3cm] at (c.north) {\textbf{c}};
\node[below, yshift=-0.3cm] at (d.south) {\textbf{d}};

\end{tikzpicture} \caption{Melia's symmetric world} \label{joe}
\end{figure}
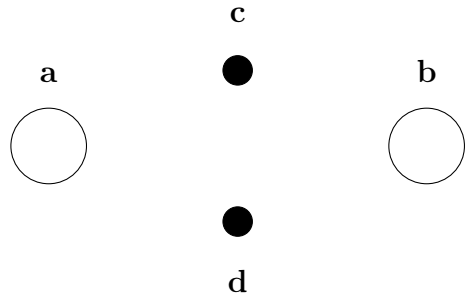

We will show, once again, that what this example illustrates is the
distinction between the definitions D1 and D3. Here we just have to be
a bit more careful with how we explicate the details of the example.

The most straightforward regimentation of Melia's example has two
domains $W$ and $B$ for the particles, and for each $t=0,1$, a
relation $\alpha _t (x,y)$ to indicate that at time $t$, $x$ is
adjacent to $y$. We add the axioms that at $t=0$, no two things are
adjacent, and that at $t=1$, the adjacency relation induces a
bijection between $W$ and $B$.\footnote{We can imagine the white and
  black particles as sitting on the vertices of a kite, and we can
  think of $\alpha$ as a metric, where at $t=0$, each white particle
  is distance $\sqrt{5}$ from each black particles; and at $t=1$, each
  white particle is distance $0$ from one black particle, and distance
  $4$ from the other.}  Since $T$ dictates that $\alpha _0$ is empty,
its presence is structurally irrelevant. Thus we drop $\alpha _0$ and
set $\alpha =\alpha _1$. Let's call the resulting theory $T$.

A model $M$ of $T$ consists of two sets, each with two elements, and a
bijection $\alpha ^M$ between them. Here the initial conditions are
just the two sets, whereas the final conditions include the bijection
$\alpha ^M$. It is clear that any two models $M,M'$ of $T$ are
isomorphic, and so $T$ automatically satisfies D1.

It will now be easy to see that $T$ does \emph{not} satisfy D3, and
the reason is that the final conditions have less symmetry than the
initial conditions. More rigorously, in a model $M$ of $T$, the
initial conditions (i.e.\ the two sets $W$ and $B$) are invariant
under any symmetry of the form $\langle f_0,f_1\rangle$, where
$f_0:W\to W$ and $f_1:B\to B$ are bijections. In contrast, the final
conditions include a bijection $\alpha ^M:W\to B$, and this bijection
is not invariant under symmetries of the form $\langle f_0,f_1\rangle$
where $f_0$ and $f_1$ do not have the same polarity (i.e.\ where $f_0$
is the identity and $f_1$ flips elements or vice versa).  Therefore,
there is a symmetry of initial conditions that does not extend to a
symmetry of models, and $T$ is D3-indeterministic.

The upshot: to understand the sense in which Melia's doubly-symmetric
world is indeterministic, we do not need to know anything about
haecceitistic differences. It is enough to see that there is a
duplication of initial conditions that does not extend to a duplication
of worlds. So this example, and others like it, only emphasizes the
virtues of ``purely formal'' definitions of determinism.

Similar remarks can be made about Hawthorne's ships. As with Belot's
and Melia's examples, his example is D1 but not D3 deterministic ---
reinforcing our claim that this distinction suffices to capture our
intuitions.

To be precise, let $T$ be the theory with a single unary predicate $p$
that says: there are exactly two things, and one of them is $p$. Here
we take $p(x)$ to mean that $x$ sinks at $t=1$. This theory is as
simple as can be imagined. A model $M$ of $T$ consists of a set with
two elements and a singleton extension for $p$. For any two models $M$
and $M'$ of $T$, there is a unique isomorphism $f:M\to M'$.  This
shows that $T$ is D1-deterministic.

But $T$ is not $D3$ deterministic. To see this, let $S=S'=\{ a,b\}$,
let $i:S\to M$ be the embedding into a model $M$ such that
$M\vDash p(a)$, i.e., $a$ sinks, and let $i':S'\to M'$ be the
embedding into a model $M'$ such that $M'\vDash p(b)$, i.e., $b$
sinks. Then $1_S:S\to S'$ is an isomorphism of initial conditions, but
there is no $g:M\to M'$ that completes the following diagram:
\[ \begin{tikzcd}
    M\arrow[r, dashed, "g"]  & M'  \\
    S \arrow[u, "i"] \arrow[r, "1_S"'] & S'\arrow[u, "i'"']
  \end{tikzcd} \] Therefore, $T$ is D3-indeterministic.\footnote{For what it is worth, Hawthorne's ships example is also indeterministic on D2, for the same reason.  Thanks to an anonymous referee for suggesting we make this explicit.}

In summary, the examples by Belot, Melia, and Hawthorne have been
thought to provide evidence for the existence of a gap between
qualitative and full determinism. What these examples actually
illustrate is the distinction between D1 and D3.

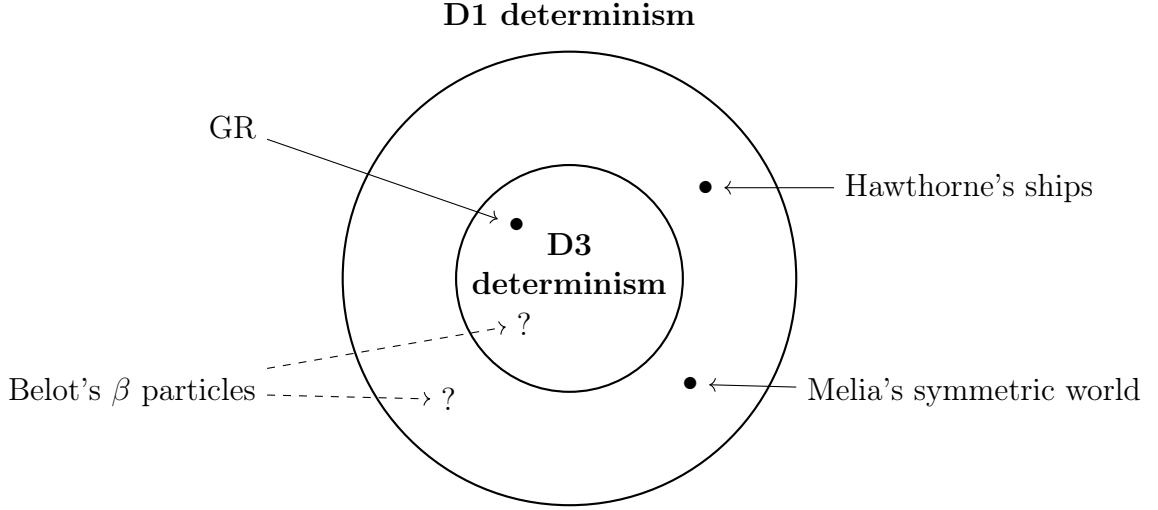
\begin{figure}[h]
  \begin{tikzpicture}
    \draw[thick] (0,0) circle (3cm);
    \node at (0,3.5) {\textbf{D1 determinism}};

    \draw[thick] (0,0) circle (1.5cm);
    \node[align=center] at (0,0.2) {\textbf{D3} \\ \textbf{determinism}};

    \node (GR) at (-0.7,0.7) {\textbullet};
    \node (Ships) at (1.8,1.2) {\textbullet};
    \node (Belot) at (-1.6,-1.6) {?};
    \node (Alt) at (-0.6,-0.6) {?};
    \node (Melia) at (1.6,-1.4) {\textbullet};

    \node[left] (GR_label) at (-4,2) {GR};
    \draw[->] (GR_label) -- (GR);

    \node[right] (Ships_label) at (3.5,1.2) {Hawthorne's ships};
    \draw[->] (Ships_label) -- (Ships);

    \node[left] (Belot_label) at (-4,-1.5) {Belot's $\beta$ particles};
    \draw[->, dashed] (Belot_label) -- (Belot);
    \draw[->, dashed] (Belot_label) -- (Alt);

    \node[right] (Melia_label) at (3,-1.5) {Melia's symmetric world};
    \draw[->] (Melia_label) -- (Melia);

  \end{tikzpicture}
  \caption{The toy examples in the literature are D3-indeterministic,
    while GR is D3-deterministic. Belot's example can be interpreted
    in two ways, one deterministic and one indeterministic.}
\end{figure}

\section{D2, D3 and Full Determinism} \label{proofs}

We have suggested that D1 captures the concept of qualitative
determinism, and the literature agrees with us on this. Where
intuitions might still clash is whether D2 or D3 captures the
strongest sense of determinism that it makes sense to ask about. We
claim that D2 is not yet full determinism, but D3 is.

A paradigm example that satisfies D2 but not D3 is Belot's: the
growing sphere of $\beta$ particles in Newtonian spacetime, at least
when formalized as consisting of two kinds of things, $\beta$
particles and spacetime points, that are \emph{not} connected by a
lawlike relation. This example satisfies D2, since any symmetry of
empty space trivially extends to a symmetry of space and its material
contents. This example does \emph{not} satisfy D3, since symmetries of
empty space do not determine symmetries of material contents. But what
does the non-uniqueness of extensions of isomorphisms have to do with
determinism? Lewis' account, read naively, says that a failure of
determinism entails the existence of two distinct worlds. But here it
is hard to see what entitles us to say that there are two distinct
worlds. After all, any two models of this toy theory are
isomorphic. The problem is not a failure of isomorphism, it is that
there are multiple isomorphisms.

We should not allow ourselves to become confused by trying to count
the number of possible worlds, or by trying, in general, to analyze
determinism in terms of static concepts. Determinism is a claim about
how things at one time are related to things at another time. As
Carnap et al.\ realized, it is more clear to describe determinism as a
property of a scientific theory: that theory posits a fixed relation
between past states and future states. If we set aside the murky talk
about identity of possible worlds, it should be clear from common
sense that Belot's example can be described in two ways: either the
relationship of the $\alpha$ particle to spatial points does determine
the relationship of the $\beta$ particles to spatial points, or it
doesn't. In the latter case --- made explicit by our toy theory with
two sorts $S$ and $B$ --- there is only one model (up to isomorphism),
but in that model, a symmetry of an initial time slice does not
determine a symmetry of a latter time slice. That's a clear sign that
that facts about the latter time are not entailed by the laws and
facts about the earlier time.\footnote{We have tacitly moved here from
  a claim about models to a claim about entailments. We suspect that
  this move can be formalized as a general version of Beth's theorem,
  showing the implicit definability entails explicit definability. In
  this case, the semantic formulation of determinism should entail a
  corresponding syntactic formulation.}

The upshot of these considerations is that neither D1 nor D2 captures
the strongest sense in which the past can determine the future. We
argue, in contrast, that D3 does.\footnote{In fact, at least for spacetime theories, more can be said about the difference between D2 and D3.  \citet{dialogue} show that
  D3 holds if and only if: D2 holds along with any one of three
  equivalent rigidity statements which are articulated in: \citet{geroch}, \citet{regarding}, and \citet{close}.}

It is not so easy to compare D3 to Full Determinism, since the latter
hasn't been given a precise, formal definition. When we tried to make
Full Determinism precise, the only way we could do so was by positing
a distinction between a de re language $\ds$ and a qualitative
sub-language $\qs$. So let's first look at the special case of
theories that come equipped with a de re language, i.e.\ names for all
objects. In that case, we can show that FD, D1, D2, and D3 are all
equivalent. But this result also shows why it is important to look at
theories \emph{without} names, where FD does not apply, because only
those theories illustrate the distinctions between D1, D2, and D3.

To make things more clear, we remind the reader of a simple point: a
$\ds$-theory $T$ is fully deterministic iff $T$ satisfies D1.

\begin{prop} For theories with sufficiently many names, FD, D1, D2,
  and D3 are equivalent. \label{collapse} \end{prop}

\begin{proof} As noted above, the definition of FD is formulated
  against a background assumption that there are two signatures
  $\qs\subset\ds$. However, FD is just D1 for a $\ds$-theory, and we
  are assuming that $T$ is such a theory. It will suffice then to show
  that D1 implies D3. Suppose then that $T$ satisfies D1, and let
  $f:U\to U'$ be an isomorphism. By D1, there is an isomorphism
  $g:M\to M'$. Since all elements of $U$ and $U'$ are named, the
  isomorphism $f:U\to U'$ is unique, hence $g|_U=f$. That is, $g$
  extends $f$. Similarly, since all elements of $M$ and $M'$ are
  named, the isomorphism $g:M\to M'$ is unique. Therefore $g$ is the
  unique extension of $f$, and $T$ satisfies D3.
\end{proof}

But what about the case where we do not have names? One natural
suggestion for extending FD to such theories is simply to add names
when necessary. Given a $\Sigma$-theory $T$, let $\Sigma ^+$ be the
expansion of $\Sigma$ to include ``sufficiently many'' names, and let
$T^+$ be the extension of $T$ to $\Sigma ^+$. We can then ask about
the relation between $T$ having property D2 or D3 and $T^+$ having
property FD (i.e., D1).  In some cases, D2 (and D3) for $T$ do imply
FD for the enriched theory. In particular, this happens for theories
of enduring objects -- that is, theories whose models' domains
coincide with the domains of their initial segments. D2 for a theory
$T$ implies that the name-enriched theory $T^+$ satisfies FD.

\begin{prop} Suppose that $T$ is a theory of enduring objects. If $T$
  satisfies D2, then $T^+$ satisfies FD. \end{prop}

\begin{proof} Suppose that $T$ satisfies D2. Let $M,M'$ be models of
  $T^+$ with initial segments $U,U'$ such that $U\de U'$. That is,
  there is a $\Sigma ^+$-isomorphism $f:U\to U'$. Since $T$ and $T^+$
  are theories of enduring objects, $|M|=|U|$ and $|M'|=|U'|$. It
  follows that $f(c^M)=c^{M'}$, for all names $c$. Since $f$ is a
  $\Sigma$-isomorphism, D2 entails that there is a
  $\Sigma$-isomorphism $g:M|_{\Sigma}\to M'|_{\Sigma}$ that extends
  $f$. Since $g$ extends $f$, $g(c^M)=f(c^M)=c^{M'}$, for all names
  $c$. Therefore, $g$ is a $\Sigma ^+$-isomorphism, and $M\de
  M'$. Therefore, $T^+$ satisfies FD.
\end{proof}

What drives this result is that if objects persists over time, then
fixing the referents of constant symbols in initial segments fixes
those referents in the entire model.  But as we have seen, not all
theories have enduring objects.  In general, we find that even D3 for
a theory $T$ comes apart from FD for an enriched theory $T^+$.  First
we now show that $T^+$ satisfying $FD$ does not imply that $T$
satisfies D3.

\begin{example}[toy Leibnizian spacetime] Let $T$ be a theory with two
  sort symbols $S_0,S_1$, and axioms that say that there are exactly
  two objects of each sort. The theory $T$ satisfies D2, but not D3,
  since an automorphism of the first sort $S_0$ extends in more than
  one way to an automorphism of the second sort $S_1$. The theory
  $T^+$ adds two constant symbols of each sort, say $a_0,b_0$ and
  $a_1,b_1$. But then for any models $M,M'$ of $T^+$, there is an
  isomorphism $f:M\to M'$, and so $M\de M'$. Therefore, $T^+$
  satisfies FD. \hfill $\Box$ \end{example}

The converse also fails: that is, a theory $T$ may satisfy D3, but the
enriched theory $T^+$ with sufficiently many names might not satisfy
FD.

\begin{example}[toy Newtonian spacetime] We now consider Newtonian
  spacetime as having a different domain $S_t$ of spatial points for
  each time, and isomorphisms $\delta _{t,t'}:S_t\to S_{t'}$ that pick
  out the preferred frame of reference. For our purposes, it suffices
  to consider a simple case with two sorts $S_0,S_1$, and a single
  function $\delta :S_0\to S_1$. Let $T$ be the theory that says there
  are exactly two elements of type $S_0$, and that $\delta$ is a
  bijection. If $M,M'$ are models of $T$, then an isomorphism
  $g:M\to M'$ consists of two bijections $g_0:S_0\to S_0'$ and
  $g_1:S_1\to S_1'$ that satisfy the compatibility condition
  $\delta ^{M'}\circ g_0=g_1\circ \delta ^M$. It follows that any
  bijection $g_0:S_0\to S_0'$ extends uniquely to a bijection
  $g:M\to M'$. Therefore, $T$ satisfies D3.

  For the name-enriched theory $T^+$, suppose that $a_0,b_0$ are
  constant symbols of sort $S_0$, and that $a_1,b_1$ are constant
  symbols of sort $S_1$. Then there is one model $M$ of $T^+$ such
  that $M\vDash \delta (a_0)=a_1$, and a non-isomorphic model $M'$ of
  $T^+$ such that $M'\vDash \delta (a_0)=b_1$. But the initial
  segments of $M$ and $M'$ are isomorphic. Therefore, $T^+$ does not
  satisfy FD. \hfill $\Box$ \end{example}

Note that in the toy Newtonian example, $T^+$ does not satisfy D3 (in
addition to not being fully deterministic). Therefore, adding names to
a D3 deterministic theory can result in a D3 indeterministic theory.
It might seem like a strike against D3 that it is not stable under the
addition of names to a theory. But that intuition is based on a false
assumption that the role of names in formal theories is the same as
the role of names in ordinary language. In a formal theory,
introducing a new name is tantamount to introducing a new property
$\varphi (x)\equiv (x=c)$. But we should expect that adding new
properties, without adding dynamical laws that govern the behavior of
those properties, could transform a deterministic theory into an
indeterministic theory.

Something similar happens with GR.  As usually formulated, GR
satisfies D3 \citep{close}.  But adding names to GR results in a
theory that does not satisfy D3, and thus is not fully deterministic
--- as shown by the hole argument.  Indeed, \citet{regarding} suggests
that one way of understanding manifold substantivalism, as described
by \citet{Earman+Norton}, is as a view on which there are additional
singular, haecceitistic facts about spacetime points that can only be
described by something like enriched --- or, in Pooley's terms,
substantivalist --- GR. The hole argument then shows this theory is
indeterministic, even on D1.  Of course, the important point is that
this enriched theory is not GR, the theory that we have good reasons
to believe at least approximately describes the structure of space and
time in our universe, but rather GR plus a great deal more structure.
Without that structure, GR is deterministic by D3, and FD does not
even apply.

\section{Interpretation revisited}

Some readers may find the analysis at the end of the previous section
unsatisfactory.  One might argue, for instance, that we have simply
misunderstood full determinism. What motivates full determinism is the
idea that there are individuals in the world, and haecceitistic facts about those individuals. A deterministic theory
ought to be able to assign them properties in an unambiguous
(deterministic) way. The names just give us a way of referring to
those individuals.  When we say that FD simply does not apply to GR, that
is not a problem for FD, it is a problem for GR!  We need enriched GR
to accurately assess whether the theory determines the properties of
individuals.  Without names, we are simply dwelling in the domain of
the qualitative.

We think this posture is wrongheaded.  But rather than argue against
it directly, we want to propose a diagnosis of where it originates,
drawing on the arguments from section \ref{sec:interpretation}.  As we
argued there, interpretation is itself formal. All anyone is doing
when they try to interpret theories is just layering models on models.
Crucially, for the present point: Tarskian semantics involves mapping
theories, with or without names, into set theory, typically understood
as a theory with names (or rather, as theory whose membership
properties allow us to uniquely individuate sets).

We suggest that the motivation for FD arises because when you
interpret a theory without names in set theory, without paying careful
attention to how the semantics works, it looks as if the ``real''
points, the ones the theory is referring to, have names.  This
apparently means we can ask about what ``determines'' what properties
those named things have.  But this instinct is a mistake.  It
illegitimately mixes two different things: the theory we are trying to
analyze, and the formal tools we use to analyze it.  Determinism for
theories is about whether initial segments of models determine the
entire model. The ``determination'' of which objects in a model carry
which properties (or names) is about an interpretation map, in the
Tarskian sense.  In other words, failures of FD are about \emph{us},
that is, about how we think about our formal semantics and how we
define interpretation maps, not about our theories or the world.

We suggest that something similar happens in many discussions of the
hole argument. Philosophers apparently mistake GR for enriched GR when
thinking about substantivalism.  Doing this is not only a mistake, it
quickly leads to incoherence. We can see this point most starkly by
considering a special sector of GR, consisting of models that satisfy
a strong asymmetry property called ``Heraclitus'' \citep{heraclitus,
  four-man}.  Spacetimes with that property are such that every point
is uniquely specified by its metrical properties (including
derivatives, i.e., curvature scalars).  Call the theory of Heraclitus
spacetimes HGR.  HGR has names, in the sense that one can uniquely
refer to points. This theory satisfies FD.  (Of course, it also
satisfies D3.)

Now consider what happens when we interpret HGR in set theory.  We
assign those named points to sets, which also have names.  But of
course, nothing in the theory can determine which (named) set we
assign to which named point. No theory can do that, because it happens
at the level of choosing an interpretation map!  What this means,
though, is that on the set theory side we now have \emph{too many
  names}.  That is, if we try to doubly-interpret this theory, as we
suggest the FD advocate would wish to, and run our analysis of
determinism on those doubly-interpreted structures, we will find that
they are not deterministic.  We claim this is a completely generic
situation that arises from layering interpretation on top of set
theory.

And it gets worse! Suppose we somehow solved this problem in the
doubly-interpreted theory, perhaps by adding laws that coordinate
between the two types of names, restoring FD.  What then?  Now we have
a new theory, with lots of redundant names, that has no expressive
resources beyond our original theory.  (Perhaps the theories are even
logically equivalent, depending on the details.)  But then we can
interpret that theory, using Tarskian semantics.  The (triply)
interpreted theory will now have three types of names --- the two
coordinated ones in our theory, plus the names of the sets on the
semantic side.  The problem will arise again.  And so on ad infinitum.
But once we see how this works, it is clearly a pseudo-problem, one
arising only because of a confusion about what is part of the theory
and what is not, what the theory \emph{should} determine and what is
merely structure added in interpretation.

\section{Conclusion}

Analytic philosophy was, in one sense, born from the idea that
philosophers should avail themselves of the tool-kit of
mathematics. We find it odd, then, that prominent analytic
philosophers have recently argued \emph{against} formal approaches,
saying things like, ``the purely formal approach is a nonstarter'',
or, ``determinism cannot be a formal property of theories.'' The
descriptive content of such claims is opaque (what is a non-formal
property?), but their illocutionary force is to recommend against
adopting the methods that distinguished analytic philosophy from the
more speculative, and less science-friendly, approaches of the
nineteenth century.  Surely this belongs among the ironies of
intellectual history.

We are also motivated by a practical concern about how to facilitate
fruitful dialogue between philosophy and the natural sciences. If
philosophers insist on making distinctions that cannot gain any
traction in scientific practice, then they will only reinforce
disciplinary boundaries that are harmful to both philosophy and the
sciences.

To be clear, we are \emph{not} arguing for a kind of science-deference
that says, ``if scientists don't regularly make some distinction, then
neither should philosophers''. For one, we recognize that scientists
might have practical reasons to blur over distinctions of genuine
metaphysical significance. One might have thought that the distinction
between qualitative and full determinism is of this sort; but our
investigation has shown that the relevant conceptual joint is not
here. The relevant conceptual joints are drawn with the help of
isomorphisms between models, and these joints clearly distinguish
which real-life scientific theories are deterministic. With this kind
of division of labor, philosophy and physics can work in tandem to
figure out whether the world is deterministic.

\printbibliography

\end{document}